\documentclass[12pt,reqno]{amsart}
\usepackage{amssymb}
\usepackage{amsfonts}
\usepackage[open,openlevel=1]{bookmark}

\usepackage{todonotes}

\theoremstyle{plain}
\newtheorem{theorem}{Theorem}[section]

\newtheorem{lemma}[theorem]{Lemma}
\newtheorem{proposition}[theorem]{Proposition}
\newtheorem{remark}[theorem]{Remark}

\newtheorem*{acknowledgement*}{Acknowledgement}

\numberwithin{equation}{section}

\newcommand{\Fourier}{\mathcal{F}}
\newcommand{\N}{\mathbb{N}}

\newcommand{\R}{\mathbb{R}}
\newcommand{\CC}{\mathbb{C}}
\newcommand{\D}{\mathbb{D}}
\newcommand{\vol}{\operatorname{vol}}

\binoppenalty=\maxdimen
\relpenalty=\maxdimen
\begin{document}
\title[A fractal uncertainty principle for Bergman spaces and analytic wavelets]
{A fractal uncertainty principle\\ for Bergman spaces and analytic wavelets}
\author{Luis Daniel Abreu}
\email{abreuluisdaniel@gmail.com}
\address{Faculty of Mathematics, University of Vienna,
Oskar-Morgenstern-Platz 1, A-1090, Vienna, Austria}

\author{Zouhair Mouayn}
\email{mouayn@gmail.com}
\address{Department of Mathematics, Faculty of Sciences and Technics (M'Ghila),
Sultan Moulay Slimane University, P.O. Box. 523, B\'{e}ni Mellal,
Morocco\\
Acoustics Research Institute, Wohllebengasse 12-14, A-1040, Vienna, Austria}

\author{Felix Voigtlaender}
\email{felix.voigtlaender@ku.de}
\address{Katholische Universität Eichstätt--Ingolstadt,
Lehrstuhl Reliable Machine Learning,
Ostenstraße 26,
85072 Eichstätt,
Germany}

\thanks{The authors were supported by the Austrian Science Fund (FWF) via
the project (P31225-N32).}

\date{\today}


%Explanation of subject classes:
%30H20: Bergman spaces and Fock spaces
%47B35: Toeplitz operators, Hankel operators, Wiener-Hopf operators
%47A30: Norms (inequalities, more than one norm, etc.) of linear operators
%47A75: Eigenvalue problems for linear operators
%42C40: Nontrigonometric harmonic analysis involving wavelets and other special systems

\maketitle

\begin{abstract}
Motivated by results of Dyatlov on Fourier uncertainty principles for Cantor sets
and by similar results of Knutsen for joint time-frequency representations
(i.e., the short-time Fourier transform (STFT) with a Gaussian window, equivalent to Fock spaces),
we suggest a general setting relating localization and uncertainty and prove,
within this context, an uncertainty principle for Cantor sets in Bergman spaces on the unit disk,
where the Cantor set is defined as a union of annuli that are equidistributed in the hyperbolic measure.
The result can be written in terms of analytic Cauchy wavelets.
As in the case of the STFT considered by Knutsen, our result consists of a two-sided
bound for the norm of a localization operator involving the fractal
dimension $\log 2/ \log 3$ in the exponent.
As in the STFT case and in Dyatlov's fractal uncertainty principle,
the (hyperbolic) measure of the dilated iterates of the Cantor set in the disk tends to infinity,
while the corresponding norm of the localization operator tends to zero.
\end{abstract}

\section{Introduction}

\subsection{Fractal uncertainty principles for the Fourier transform}

The uncertainty principle is a collection of statements in harmonic
analysis, each of them quantifying in some form the fundamental duality
between a function $f$ and its Fourier transform $\Fourier f$, which
prevents both representations from being ``simultaneously concentrated in small sets'' \cite{RT}.
Let us consider the Fourier transform $\Fourier : L^{2}(\R) \to L^{2}(\R)$ given by
\[
  \Fourier f(\xi)
  = \widehat{f} (\xi)
  = (2\pi)^{-1/2} \int_{\R} e^{-i x \xi} f(x) dx
  \quad \text{for} \quad
  f \in L^1(\R) \cap L^2(\R) ,
\]
and the $h$-dilated Fourier transform $\Fourier_{h} : L^{2}(\R) \to L^{2}(\R)$
\[
  \Fourier_{h} f(\xi)
  = h^{-1/2} \cdot \Fourier f (h^{-1} \xi)
  .
\]
Several mathematical manifestations of the uncertainty principle consist of
bounds on the norm of the operator which concentrates the energy of $f$ in
a set $X$ and the energy of $\mathcal{F}f$\ in a set $Y$.
Following Dyatlov's definition in \cite{Dyatlov}, one can resort to the $h$-dilated
Fourier transform, and declare a pair of real $h$-dependent sets $X \subset \R$ and $Y \subset \R$
to satisfy an uncertainty principle with exponent $\beta > 0$ if, as $h \to 0$,
\begin{equation}
  \big\Vert 1_{X} \Fourier_{h} 1_{Y} \big\Vert_{op}
  = O(h^{\beta})
  \rightarrow 0
  ,
\label{bound}
\end{equation}
where the operator $T = 1_X \Fourier_h 1_Y$ acts via $T f = 1_X \cdot \Fourier_h [1_Y \, f]$
with $1_X$ denoting the indicator function of the set $X$.
For instance, if $X,Y = [0,h]$ then the Hölder inequality,
combined with the estimate $\| \Fourier_h f \|_{L^\infty} \leq h^{-1/2} \, \| f \|_{L^1}$
gives the uncertainty principle
\begin{equation}
  \big\Vert 1_{X}\mathcal{F}_{h}1_{Y} \big\Vert_{op}
  = O(h^{\frac{1}{2}})
  \rightarrow 0
  .
  \label{first}
\end{equation}
Given $R > 0$, set $h=R^{-2}$.
Then
\begin{equation}
  \big\Vert 1_{X}\mathcal{F}_{h}1_{Y} \big\Vert_{op}
  = \left\Vert 1_{X/\sqrt{h}} \Fourier 1_{Y/\sqrt{h}}\right\Vert_{op}
  = \left\Vert 1_{RX} \Fourier 1_{RY} \right\Vert_{op}
  .
  \label{eq:FourierUncertaintyRescaling}
\end{equation}
Now, \eqref{first} becomes an uncertainty principle for the Fourier
transform and the dilated sets $RX$ and $RY$: as $R\rightarrow \infty $,
\[
  \big\Vert 1_{RX} \mathcal{F} 1_{RY} \big\Vert_{op}
  = O(R^{-1})
  \to 0
  .
\]
Moreover, as $R \rightarrow \infty $, \eqref{bound} becomes
\[
  \left\Vert 1_{X}\mathcal{F}_{h}1_{Y}\right\Vert _{op}
  = \left\Vert 1_{X/\sqrt{h }}\mathcal{F}1_{Y/\sqrt{h}}\right\Vert _{op}
  = \left\Vert 1_{RX}\mathcal{F} 1_{RY}\right\Vert _{op}
  = O(R^{-2\beta })
  \to 0
  .
\]
We want to emphasize the following aspect of \emph{the fractal uncertainty principle}
\cite{Dyatlov,DB,FUP}: it covers situations where $RX$
and $RY$ are close to having a fractal structure (they depend on $R$ and
approach fractals when $R \rightarrow \infty$), where their volume approaches
$\infty$ as $R \to \infty $, but nevertheless $RX$ and $RY$
satisfy an uncertainty principle (their operator norm decays like $O(R^{-\beta })$
for some $\beta > 0$).
This can roughly be described by saying that
\emph{no function can be localized close to a fractal set in both time and frequency}.

For illustration and motivation, consider the dilated real Cantor set
$C(R; \mathbb{R})\subset [0,R]$ defined as
\[
  C(R;\mathbb{R})
  = \bigcap_{n=1}^{\infty}
      C_{n}(R;\mathbb{R})
  = \bigcap_{n=1}^{\infty}
      RC_{n}
  ,
\]
where $C_{n}$ is the $n$-th step in the iterative construction of the Cantor set,
where $C_{n+1}$ is obtained from $C_n$ by noting that $C_n$ is a finite union of intervals,
from each of which one removes the middle third to obtain $C_{n+1}$;
for instance, $C_0 = [0,1]$, $C_1 = [0,\frac{1}{3}] \cup [\frac{2}{3}, 1]$, etc.
The set $C_{n}(R; \mathbb{R})$ thus consists of $2^{n}$ disjoint intervals $I_{n,j} \subset [0,R]$,
each with measure $\vol(I_{n,j}) = 3^{-n}R$.
Taking $n$ such that $\vol(I_{n,j})\asymp 1/R$, then $3^{n/2}\asymp R$ and
\[
  \vol(C_{n}(R;\mathbb{R}))
  = 2^{n} \vol(I_{n,j})
  \asymp R^{-1+2\frac{\ln 2}{\ln 3}}
  \to    \infty
  ,
\]
as $R \to \infty $.
Using this volume bound and Hölder's inequality as in \eqref{first}, leads to
\[
  \big\Vert 1_{C_n(R;\R)} \Fourier 1_{C_n(R;\mathbb{R})} \big\Vert_{op}
  \lesssim R^{-1 + 2 \frac{\ln 2}{\ln 3}}
  \to \infty
  ,
\]
therefore not enough to assure that $X = Y = C(R;\R)$ satisfy an
uncertainty principle according to Dyatlov's definition
(see also Equation~\eqref{eq:FourierUncertaintyRescaling}).
However, from Example 2.6 and Theorems 2.12 and 2.13 in \cite{Dyatlov} it follows that there exists an
exponent $\beta > 0$ such that, as $R \to \infty $, \eqref{bound}
holds for $RX,RY = C_{n}(R;\mathbb{R})$ and $h = R^{-2}$.
Since $R \asymp 3^{n/2}$, as $n \rightarrow \infty$, also $R \to \infty$ and
Equation~\eqref{eq:FourierUncertaintyRescaling} implies that
\begin{equation}
  \big\Vert 1_{C_{n}(R;\R)}\mathcal{F}1_{C_{n}(R;\R)} \big\Vert_{op}
  = O(R^{-2\beta })
  \to 0
  .
  \label{FUP1}
\end{equation}

\subsection{Uncertainty and localization}
\label{sub:IntroUncertaintyLocalization}

The time-frequency localization operator of the previous paragraph suggests
a construction in a general Lebesgue space $L^{2}(\Lambda)$
(with $\Lambda $ being a metric measure space).
If $L^{2}(\Lambda)$ has a reproducing kernel $K(z,w)$ (for a subspace of $L^2(\Lambda)$),
one can define a localization/Toeplitz
operator $P_{\Omega}$ mapping the function $f\in L^{2}(\Lambda)$ to a smooth function
$P_{\Omega} f \in L^{2}(\Lambda)$ essentially concentrated in a bounded region $\Omega \subset \Lambda$.
The operator $P_{\Omega}$ is explicitly defined as
\[
  (P_{\Omega }f)(w)
  = \int_{\Omega} f(z) \overline{K(z,w)} d\mu (z)
  .
\]
In this context, we say that $\Omega = \Omega (R)$ ($R > 0$) satisfies an uncertainty
principle if, as $R \to \infty$,
\[
  \left\Vert P_{\Omega} \right\Vert_{op}
  = O(R^{-\beta })
  \quad \text{for some } \beta > 0
  .
\]
As in the previous paragraph, one can look for bounds of $\left\Vert P_{\Omega }\right\Vert_{op}$
when $\Omega $ is a fractal set.
In \cite{KnutsenLocalizationPaper}, the case of the Fock space setting has been
considered, in the equivalent formulation provided by the short-time-Fourier
transform \cite[Chapter 3]{Charly}.
For the analogy with our results, it will be convenient to rephrase the results in the Fock space.
This corresponds to the choice of the measure $d\mu (z)=e^{-\pi \left\vert z\right\vert ^{2}}dz$
on $\Lambda = \CC$ (where $dz$ is Lebesgue measure), of the kernel
\[
  K(z,w)
  = K_{Fock}(z,w)
  = e^{\pi z\overline{w}}
  ,
\]
and of $\Omega \subset \CC$ as the $n$-th iterate $C_n (R;\CC)$ of the \emph{planar Cantor set},
i.e.,
\begin{equation}
  C_{n}(R;\CC)
  := \{ z \in \CC \colon |z|^{2} \in C_{n}(R^{2};\R)\}
  .
  \label{eq:KnutsenCantor}
\end{equation}
The set $C_n(R; \CC)$ is a disjoint union of $2^{n}$ annuli $I_{n,j}^{\CC}$,
each of measure $\vol(I_{n,j}^{\mathbb{C}}(R^{2})) \!=\! \frac{\pi \,R^{2}}{3^{n}}$,
so that for each annulus we consider a $1/3^{n}$ part of the initial disk with area $\pi \,R^{2}$.
For this measure to be well distributed among $[ 0,\pi R^{2}] $, one takes
$\vol(I_{n,j}^{\CC}(R^{2})) = 1/(\pi \,R^{2})$, yielding $(\pi \,R^{2})^{2}\asymp 3^{n}$ and
\begin{equation*}
  \vol(C_{n}(R;\CC))
  \asymp R^{-2 + 4 \frac{\ln 2}{\ln 3}}
  \to \infty
\text{,}
\end{equation*}
as $R \to \infty$.
Moreover, as $R \to \infty $, it is proven in \cite[Corollary 4.1]{KnutsenLocalizationPaper} that
\[
  \left\Vert P_{C_{n}(R;\mathbb{C})}\right\Vert_{op}
  \asymp R^{-2+2\frac{\ln 2}{\ln 3}}\rightarrow 0
  .
\]
Thus, in the Fock case, the measure of the dilates of the iterated Cantor
set tends to $\infty $, while the norm of the operator tends to zero.

Small operator norms facilitate recovery in signal analysis problems \cite{AS}.
The operator norm of $P_{\Omega}$ for the Fock case discussed above
is maximized when $\Omega$ is a disk \cite{FK},
while there exist sets $\Omega $ with infinite Lebesgue
measure such that the operator norm of $P_{\Omega }$ is arbitrarily small \cite{Galbis}.
In the wavelet case, the operator norm is maximized when the localization domain
is a pseudohyperbolic disk \cite{RamosTilli}.

We will show in this paper that, in the case of the disk, we have a similar
situation: one can define \emph{a Cantor set in the disk, whose hyperbolic
measure tends to infinity}, and an associated Toeplitz operator,
\emph{whose operator norm tends to zero}.

Our contributions are organized as follows:
A disk version of the Cantor set is considered in the next section.
In the same section, the main result on the fractal uncertainty principle on Bergman spaces
is stated and translated to the language of analytic wavelets.
More details and proofs are given in the following two sections,
with the most technical estimates delegated to the last section of the paper.

\section{Fractal uncertainty principles for the Bergman space}

In this paper, we consider the reproducing kernel
of the \emph{weighted analytic Bergman space}
associated to the measure $dA_{\alpha }(z) = 2\alpha \, (1 - |z|^{2})^{2\alpha - 1} dA(z)$
on the disk $\D = \{ z \in \CC \colon |z| < 1 \}$,
where $\alpha \in (0,\infty)$ and where $dA(z) = \frac{dz}{\pi}$, with
$d z$ denoting the planar Lebesgue measure.
As shown in  \cite[Pages 4 and 5]{HKZ}, this reproducing kernel is explicitly given by
\begin{equation}
  \mathcal{K}_{\mathbb{D}}^{\alpha }(z,w)
  =\frac{1}{(1-z\overline{w})^{2\alpha + 1}}
  = \sum_{n=0}^{\infty} e_{n}^{\alpha }(z) \overline{e_{n}^{\alpha }(w)}
  ,
  \label{repBergman}
\end{equation}
where
\[
  e_{n}^{\alpha}(z)
  = \sqrt{\gamma _{n}}z^{n},
  \quad \text{and} \quad
  \gamma_{n}
  = \gamma_{n}^{\alpha}
  := \frac{\Gamma (n+1+2\alpha)}{\Gamma (n+1)\Gamma (1 + 2\alpha)}
  =\frac{1}{2\alpha} [B(n+1, 2\alpha)]^{-1}
  .
\]
Given this reproducing kernel, we consider the associated localization operator
\begin{equation}
  (P_{C_{n}(R;\mathbb{D})}^{(\alpha )}f)(w)
  = \int_{C_{n}(R;\D)}
      f(z) \, \overline{\mathcal{K}_{\mathbb{D}}^{\alpha }(z,w)}
    dA_{\alpha }(z)
  ,
  \label{eq:OurLocalizationOperator}
\end{equation}
where the fractal localization region is now the following disk version of
the iterates of the Cantor set
\begin{equation}
  C_{n}(R;\mathbb{D})
  := \left\{
       z \in \D
       \,\,:\,\,
       \frac{|z|^{2}}{1-|z|^{2}}\in C_{n}(R;\R)
     \right\}
  .
  \label{eq:DiscCantorSet}
\end{equation}
We will show (see Proposition~\ref{prop:DiskCantorGeometricProperties})
that\ $C_{n}(R;\mathbb{D})$ is a disjoint union of $2^{n}$
annuli $D_{\ell}^{(n)}(0,R)$, each of hyperbolic measure
$\mu_{\D}\bigl( D_{\ell}^{(n)}(0,R)\bigr) =\frac{R}{3^{n}}$.
Here, the hyperbolic measure $\mu_{\D}$ is given by
\begin{equation}
  \mu_{\D} (M)
  := \int_{M}
       (1 - |z|^2)^{-2}
     \, d A (z)
  \quad \text{for } M \subset \D \text{ Borel measurable}
  ,
  \label{eq:HyperbolicMeasure}
\end{equation}
so that the hyperbolic measure of the disk $D(0,r) = \{ w \in \CC \colon |w| < r \}$
with $r \in [0,1)$ is given by
\begin{equation}
  \mu_{\mathbb{D}}(D(0,r))
  = \int_{D(0,r)}
      (1 - |z|^{2})^{-2}
    d A(z)
  = 2 \int_{0}^r
        \frac{s}{(1 - s^2)^2}
      \, d s
  = \frac{r^{2}}{1-r^{2}}
  .
  \label{eq:HyperbolicMeasureOfDisk}
\end{equation}

Thus, as in the previous examples, for this measure to be well distributed among
$[0, R] $ we take $n$ such that $\mu_{\D}\bigl( D_{\ell}^{(n)}(0,R)\bigr) \asymp 1/R$,
leading to $3^{n}\asymp R^{2}$ and to $R^{2\frac{\ln 2}{\ln 3}}\asymp 2^{n}$, so that
\[
  \mu_{\D}\bigl(C_{n}(R;\mathbb{D})\bigr)
  = R \cdot \left( \frac{2}{3}\right)^{n}
  \asymp R^{2\frac{\ln 2}{\ln 3}-1}
  \to \infty
  ,
\]
as $R,n \to \infty$.
So far, everything is perfectly tuned with our model Fourier and time-frequency/Fock cases.
However, the analogue of the conditions (\ref{FUP1}) only holds in the asymp\-to\-tic case.
The bounds on the non-asymptotic case depend on the size of $R$.

\begin{theorem}\label{thm:CantorDiskLocalizationOperatorNorm}
Given $\alpha \in (0,\infty )$, there are constants $0 < C_{1}\leq C_{2} < \infty$
(which only depend on $\alpha $) such that the operator norm of the time-scale
localization operator $P_{C_{n}(R;\mathbb{D})}^{(\alpha)}$ satisfies
for all $n \in \N$ and $R \in (0,\infty)$ the estimate
\[
  C_{1}
  \begin{cases}
    \left( \frac{2}{3}\right) ^{n}R,                         & \text{if }0<R\leq 1 \\
    \left( \frac{2}{3}\right) ^{n}R^{1-\frac{\ln 2}{\ln 3}}, & \text{if }1\leq R\leq 3^{n} \\
    1,                                                       & \text{if }R\geq 3^{n}
  \end{cases}
  \leq \left\Vert P_{C_{n}(R;\mathbb{D})}^{(\alpha )}\right\Vert _{op}
  \leq C_{2}
       \begin{cases}
         \left( \frac{2}{3}\right) ^{n}R,                         & \text{if }0<R\leq 1, \\
         \left( \frac{2}{3}\right) ^{n}R^{1-\frac{\ln 2}{\ln 3}}, & \text{if }1\leq R\leq 3^{n}, \\
         1,                                                       & \text{if }R\geq 3^{n}.
       \end{cases}
\]
Furthermore, if $R$ is chosen so that $R^{2} \asymp 3^{n}$, then
\[
  \left\Vert P_{C_{n}(R;\mathbb{D})}^{(\alpha )}\right\Vert_{op}
  \asymp \left( \frac{2}{3}\right) ^{\frac{n}{2}}
  \asymp R^{\frac{\ln 2}{\ln 3}-1}
  \to 0
  ,
\]
as $R\rightarrow \infty $.
\end{theorem}

We note that the term
\[
  \delta
  = \delta _{C(R;\mathbb{R})}
  = \frac{\ln 2}{\ln 3}
\]
appearing in the exponent of $R^{-1+\frac{\ln 2}{\ln 3}}$ is
the Hausdorff dimension of the Cantor set.
The uncertainty principles in \cite{Dyatlov} consider more general fractal sets
and the results are obtained in terms of their Hausdorff dimensions.
See \cite{KN2,KN3} for new developments in this direction in the planar case.
As in the joint time-frequency case, following the suggestion
in the comments after \cite[Corollary 4.1]{KnutsenLocalizationPaper},
where the bound $R^{-2+2\delta }$ is obtained for the planar Cantor set,
this opens interesting problems, if one considers
more general fractal sets and seeks bounds of the associated Toeplitz
operator in terms of the Hausdorff dimension of the sets.

\subsection{Fractal uncertainty principle for analytic wavelets}

In this section we outline how our result can be written in terms of analytic wavelets.
We will use the basic notation for $\mathcal{H}^{2}(\CC^+)$,
the Hardy space in the upper half plane $\CC^+$, as the space of analytic functions
$f : \CC^+ \to \CC$ such that
\[
  \sup_{0 < s < \infty} \,\,
    \int_{-\infty}^{\infty}
      |f(x+is)|^{2}
    dx
  < \infty
  .
\]
To simplify the computations it is often convenient to use the equivalent
definition (since the Paley-Wiener theorem (see \cite{DGM} or \cite[Theorem~19.2]{RudinRCA})
shows up to canonical identifications that
$\Fourier (\mathcal{H}^{2}(\mathbb{C}^{+})) = L^{2}(0,\infty)$)
\[
  \mathcal{H}^{2}(\mathbb{C}^{+})
  = \left\{ f\in L^{2}(\mathbb{R}):(\mathcal{F} f)(\xi )
  = 0
  \text{ for almost all } \xi < 0\right\}
  .
\]
The \emph{wavelet transform} of a function $f \in \mathcal{H}^{2}(\CC^{+})$
with \emph{mother wavelet} $\psi \in \mathcal{H}^{2}(\CC^{+})$,
such that its \emph{admissibility constant}
$C_{\psi} = 2 \pi \cdot \left\Vert \Fourier \psi \right\Vert_{L^{2}(\R^+, t^{-1} dt)}^{2}$ is finite,
is defined as
\begin{equation}
  W_{\psi }f(z)
  := \int_{\mathbb{R}}
       f(t) \overline{s^{-\frac{1}{2}} \psi(s^{-1}(t-x))}
     dt
   = \sqrt{s}
     \int_{0}^{\infty}
       \widehat{f}(\xi) \overline{\widehat{\psi}(s\xi)} e^{ix\xi }
     d\xi ,
  \quad z = x + is \in \CC^+
  .
  \label{eq:WaveletTransform}
\end{equation}
The \emph{analytic wavelets} are the functions $\psi _{0}^{\alpha }$ defined
via their Fourier transforms by
\[
  (\mathcal{F}\psi _{0}^{\alpha}) (\xi)
  = \xi^{\frac{\alpha}{2}}
    \cdot e^{-\xi}
    \cdot 1_{(0,\infty)} (\xi)
  ,
  \quad \xi \in \R
  .
\]
As proven recently in \cite{AnalyticWavelet}, $W_{\psi} f(z)$ leads to
analytic (Bergman) phase spaces only for this special choice of $\psi $ (up to a phase factor).
Nevertheless, it is customary to refer to $W_{\psi }f(z)$ as the
analytic wavelet transform, due to the discard of negative frequencies.
We will write
\[
  d\mu^{+}(z)
  = \left( \mathrm{Im} z \right)^{-2} dz
  ,
\]
where $d z$ is the Lebesgue measure on\ $\CC^{+}$.
The orthogonality relations for the wavelet transform
\begin{equation}
  \int_{\CC^{+}}
     W_{\psi _{1}}f_{1}(z)
     \overline{W_{\psi _{2}}f_{2}(z)}
   d\mu ^{+}(z)
   = 2 \pi
     \cdot \left\langle \Fourier \psi_{1}, \Fourier \psi_{2} \right\rangle_{L^{2}(\R^{+},t^{-1}dt)}
     \cdot \left\langle f_{1}, f_{2}\right\rangle_{\mathcal{H}^{2}(\CC^{+})}
  \label{ortogonalityrelations}
\end{equation}
are valid for all $f_{1},f_{2}\in \mathcal{H}^{2}(\CC^{+})$
and all admissible $\psi_{1}, \psi_{2} \in \mathcal{H}^{2}(\CC^{+})$; see \cite[Proposition~2.4.1]{Dau}.
Then, setting $\psi_{1} = \psi_{2} = \psi$ and $f_{1}=f_{2}$ in \eqref{ortogonalityrelations},
gives
\[
  \int_{\CC^{+}}
     \left\vert W_{\psi }f(z)\right\vert^{2}
  d\mu^{+}(z)
  = C_{\psi} \cdot \left\Vert f \right\Vert_{\mathcal{H}^{2}(\CC^{+})}^{2}
  ,
\]
showing that the continuous wavelet transform provides an (up to a constant factor) isometric
inclusion
\[
  W_{\psi} : \quad
  \mathcal{H}^{2}(\CC^{+}) \to L^{2}(\CC^{+}, \mu^{+})
  .
\]
Writing $z = x + i s$ and setting $\psi _{1} = \psi _{2} = \psi$ and $f_{2} = \pi(z) \psi$,
where
\[
  [\pi (z)]\psi (t)
  := s^{-\frac{1}{2}} \cdot \psi (s^{-1}(t-x))
  ,
\]
in \eqref{ortogonalityrelations}, then for every $f\in \mathcal{H}^{2}(\CC^{+})$, one has
\[
  W_{\psi }f(z)
  = \frac{1}{C_{\psi}}
    \int_{\CC^{+}}
       W_{\psi} f(w)
       \langle \pi (w)\psi ,\pi (z)\psi \rangle
    d\mu^{+}(w),
  \quad z \in \CC^{+}
  .
\]
Thus, the range of the wavelet transform
\[
  \mathcal{W}_{\psi}
  := \big\{
       F \in L^{2}(\CC^{+}, \mu^{+})
       \,\,:\,\,
       F = W_{\psi }f, \, f\in \mathcal{H}^{2}(\CC^{+})
     \big\}
\]
is a closed subspace of $L^{2}(\CC^{+}, \mu^+)$, with reproducing kernel
\[
  k_{\psi }(z,w)
  = \frac{1}{C_{\psi}}
    \langle \pi (w)\psi ,\pi (z)\psi \rangle_{\mathcal{H}^{2}(\CC^{+})}
  = \frac{1}{C_{\psi}}
    W_{\psi} \psi ({w}^{-1}z),
  \quad \text{and} \quad
  k_{\psi}(z,z)
  = \frac{\Vert \psi \Vert _{2}^{2}}{C_{\psi }}
  .
\]
Here, the multiplication (and inversion) on $\CC^+$ is \emph{not} the usual multiplication
inherited from $\CC$, but stems from identifying $\CC^+ \cong \R \times (0,\infty)$
with the $a x + b$ group, so that $(x + i s) (y + i v) = x + s y + i s v$.

For $z = x + is, w = y + iv\in \CC^{+}$, the  kernel $k_{\psi_{0}^{2\alpha}}(z,w)$ is given by%
\footnote{To see this, first note that $\int_0^\infty x^{z-1} e^{-s x} \, d x = \frac{\Gamma(z)}{s^z}$
for all $s, z \in \CC$ with $\operatorname{Re} s, \operatorname{Re} z > 0$;
this can be seen by keeping $z$ fixed and verifying the identity for $s \in (0,\infty)$;
this is enough, since both sides of the identity are holomorphic functions.
Combining this formula with the definition of $\psi_0^{2 \alpha}$ and the right-hand side
of Equation~\eqref{eq:WaveletTransform} easily implies
\(
  W_{\psi_0^{2\alpha}} \psi_0^{2 \alpha} (x + i s)
  = s^{\frac{1}{2} + \alpha} \int_0^\infty \xi^{2 \alpha} e^{-(1 + s - i x) \xi} \, d \xi
  = s^{\frac{1}{2} + \alpha} \frac{\Gamma(2 \alpha + 1)}{(1 + s - i x)^{2 \alpha + 1}}
  .
\)
Furthermore, directly from the definition of $C_\psi$ and of $\psi_0^{2\alpha}$, we see
\(
  C_{\psi_0^{2\alpha}}
  = 2 \pi \int_0^\infty \xi^{2 \alpha} e^{-2 \xi} \frac{d \xi}{\xi}
  = \frac{\pi}{2^{2\alpha - 1}} \Gamma(2 \alpha)
  .
\)
From this, Equation~\eqref{eq:repr-kernel} follows easily.
}
\begin{equation}
  k_{\psi_{0}^{2\alpha}}(z,w)
  = \frac{1}{C_{\psi _{0}^{2\alpha }}}W_{\psi_{0}^{2\alpha}} \psi _{0}^{2\alpha}({w}^{-1}\cdot z)
  = \frac{2^{2\alpha} \alpha}{\pi}
    \cdot \left(
            \frac{\sqrt{\mathrm{Im}\,z\,\mathrm{Im}\,w}}{-i(z-\overline{w})}
          \right)^{2\alpha + 1}
  .
  \label{eq:repr-kernel}
\end{equation}
This is a multiple of the reproducing kernel of the Bergman space
\[
  A^{2\alpha - 1}({\mathbb{C}}^{+})
  := \left\{
       f : \CC^{+} \to \CC \text{ holomorphic}
       \,\,:\,\,
       \int_{\CC^{+}}
         |f(z)|^{2} \cdot \text{Im}(z)^{2\alpha - 1}
       d z
       <\infty
     \right\}
  .
\]
Therefore,
\(
  (\mathrm{Im} \, \cdot)^{-\alpha -1/2} \, \mathcal{W}_{\psi _{0}^{2\alpha }}
  : \mathcal{H}^{2}(\CC^{+}) \to A^{2\alpha -1}(\CC^{+})
\)
is an (up to a constant factor) isometric inclusion.
Moreover, $A^{2\alpha -1}(\CC^{+})$ is conformally equivalent to the Bergman space
$A^{2\alpha -1}({\mathbb{D}})$ on the unit disk under the transformation
\[
  z \in {\mathbb{C}}^{+} \mapsto \xi (z) = \frac{z-i}{z+i} \in \mathbb{D}
  ;
\]
see \cite{DGM}.
It follows that $P_{C_{n}(R;\mathbb{D})}^{(\alpha )}$ can be unitarily mapped
to (a constant multiple of) the operator
\begin{equation}
  \bigl(P_{\xi^{-1}(C_{n}(R;\mathbb{D}))}f\bigr)(w)
  = \int_{\xi ^{-1}(C_{n}(R;\mathbb{D}))}
      f(z) \cdot \overline{k_{\psi _{0}^{2\alpha }}(z,w)}
    dz
  .
  \label{DP}
\end{equation}
Thus, Theorem~\ref{thm:CantorDiskLocalizationOperatorNorm}
is equivalent to an uncertainty principle for wavelet
representations on the Cantor set of ${\mathbb{C}}^{+}$ defined by $\xi^{-1}(C_{n}(R;\mathbb{D}))$.
The analyzing wavelets $\psi_{0}^{2\alpha}$ are
the only ones leading to an analytic structure \cite{AnalyticWavelet}.
Considering more general classes of analyzing wavelets, as those leading to a
polyanalytic decomposition in \cite{VasiBergman,SF,Hun11}, or the slightly
different ones connected to Maass forms and hyperbolic Landau levels \cite{Mouayn,Maass},
may be a natural extension of the problem we have considered here.
The Toeplitz operator \eqref{DP} has been first considered by
Daubechies and Paul \cite{dapa88} for $n = 0$, i.e., for $C_0(R; \D)$
The approach based on double orthogonality that we use is due to Seip \cite{SeipDoubleOrt}
and provides some insight on why the approach using circular symmetric sets in the plane and in the
disk considerably simplifies the problem.
With square time-frequency regions, as required by the Fourier transform approach of Dyatlov,
one has no access to explicit eigenvalue formulas.
The general eigenvalue problem of Gabor and wavelet localization operators
has been considered in \cite{DeMarie}.

\section{The Cantor set}
\label{sec:CantorSet}

\subsection{The Cantor set in the line}
\label{sub:CantorSetReal}

The usual Cantor set $C \subset [0,1]$ is defined as $C \!=\! \bigcap_{n=1}^\infty C_n$,
where each of the sets $C_n$ is a finite union of closed intervals
which are iteratively constructed by the usual operation
of ``removing the middle third'' of each of the intervals;
see e.g.\ \mbox{\cite[Section~2.44]{BabyRudin}}.
For instance,
\[
  C_{0} = [0,1],
  \qquad
  C_{1} = [0,\tfrac{1}{3}]\cup [ \tfrac{2}{3},1],
  \quad \text{and} \quad
  C_{2} = [0,\tfrac{1}{9}]
          \cup [\tfrac{2}{9}, \tfrac{1}{3}]
          \cup [\tfrac{2}{3},\tfrac{7}{9}]
          \cup [\tfrac{8}{9},1]
  .
\]

Now, given $R > 0$, we consider the dilated real Cantor set $C(R;\R) = R \cdot C \subset [0,R]$.
Therefore, we have
\[
  C(R;\mathbb{R})
  = \bigcap_{n=1}^{\infty }C_{n}(R;\mathbb{R})
  ,
\]
where
\begin{equation}
  C_{n}(R;\mathbb{R})
  := R\cdot C_{n}
  \quad \text{with} \quad
  C_{n} = \biguplus_{a\in \{0,2\}^{n}}
            \left[\,
              \sum_{j=1}^{n}\frac{a_{j}}{3^{j}} \,, \quad
              3^{-n}+\sum_{j=1}^{n}\frac{a_{j}}{3^{j}}
            \,\right]
  .
  \label{eq:RealLineCantorSet}
\end{equation}

To simplify the notation, let $\Omega ^{(n)} := \{0,2\}^{n}$,
and for $a = (a_{1}, \dots, a_{n}) \in \Omega ^{(n)}$ define
\[
  b_{a}
  = b_{a}^{(n)}
  := \sum_{j=1}^{n}
       \frac{a_{j}}{3^{j}},
\]
so that
\begin{equation}
  C_{n}
  = \biguplus_{a\in \{0,2\}^{n}}
      I_a^{(n)}
  \qquad \text{with} \qquad
  I_a^{(n)} := \left[\, b_{a}^{(n)} \,, 3^{-n}+b_{a}^{(n)} \right]
  .
  \label{eq:RealCantorDefinition}
\end{equation}
For analyzing the size of the endpoints $b_{a}$, it is convenient to
introduce for $a\in \Omega ^{(n)}\setminus \{0\}$ the \emph{order} of its index as
\[
  \omega_{a}
  := \min \big\{ \ell \in \{1,\dots ,n\}\colon a_{\ell }\neq 0 \big\}
  ,
\]
which gives rise to the restricted index sets
\[
  \Omega_{m}^{(n)}
  := \left\{
       a\in \Omega^{(n)}\setminus \{0\}
       \quad \colon \quad
       \omega _{a}=m
     \right\}
  \quad \text{for}\quad m\in \{1,\dots ,n\}
  .
\]
It is straightforward to verify that $\Omega _{m}^{(n)} = \{0\}^{m-1}\times \{2\} \times \{0,2\}^{n-m}$,
and hence
\begin{equation}
  \#\Omega _{m}^{(n)}
  = 2^{n-m}
  .
  \label{eq:SpecialIndexSetCardinality}
\end{equation}
We will frequently use the estimate
\begin{equation}
  \frac{2}{3^{m}}
  \leq b_{a}
  =    \sum_{j=m}^{n} \frac{a_{j}}{3^{j}}
  \leq \frac{2}{3^{m}}
       \sum_{\ell=0}^{\infty}
         3^{-\ell }
  =    \frac{3}{3^{m}}
  \qquad \forall \,a\in \Omega_{m}^{(n)}.
  \label{eq:LowerCantorBoundAsymptotic}
\end{equation}

\subsection{The Cantor set in the plane}
\label{sub:CantorRemark}

The key properties of the Cantor-type set defined in the complex plane in
\cite{KnutsenLocalizationPaper} will be preserved in the construction of the
next subsection (modulo the required adaptations to the disk).
First, since $C_{n}(R^{2};\R)$ is a disjoint union of $2^{n}$ intervals,
the set $C_{n}(R; \CC)$ from \eqref{eq:KnutsenCantor} is a disjoint union of $2^{n}$ annuli.
Furthermore, the condition $|z|^{2}\in C_{n}(R^{2};\mathbb{R})$ is chosen since it ensures
that all of the $2^{n}$ annuli have the same Lebesgue measure.

\subsection{The Cantor set in the disk}
\label{sub:CantorSetDisk}

Our goal is to define the $n$-th iterate of the Cantor-type set in the disk
such that it is a union of $2^{n}$ disjoint annuli with
the same hyperbolic measure, where we saw in Equation~\eqref{eq:HyperbolicMeasureOfDisk} that
\[
  \mu_{\D} (D(0,r))
  = \frac{r^2}{1 - r^2}
  \quad \text{for} \quad
  D(0,r) = \{ w \in \CC : |w| < r\}
  \text{ and } r \in [0,1).
\]
This suggests defining the $n$-th Cantor set as the set of all $w \in \D$
such that $\varphi (|w|) \in C_{n}(R;\R)$, where
\[
  \varphi : \quad
  [0,1)\rightarrow [0, \infty), \quad
  r \mapsto \varphi(r)
    =       \frac{r^{2}}{1-r^{2}}
  .
\]
More explicitly, as in \eqref{eq:DiscCantorSet}, we define
\[
  C_{n}(R;\D)
  := \left\{
        w \in \D
        \,\,\colon\,\,
        \frac{|w|^{2}}{1-|w|^{2}} \in C_{n}(R;\R)
     \right\}
  .
\]
We next show that this construction indeed yields a set $C_{n}(R; \D)$
that behaves similarly to the set $C_{n}(R;\CC)$ from
\cite{KnutsenLocalizationPaper} discussed in Sections~\ref{sub:CantorRemark}
and \ref{sub:IntroUncertaintyLocalization}.
For $a \in \Omega^{(n)}$, let us write
\begin{equation}
  D_{a}^{(n)}(0,R)
  := \left\{
       w\in \mathbb{D}
       \,\,:\,\,
       \varphi(|w|) \in R \cdot I_{a}^{(n)}
     \right\}
  \quad \text{with} \quad
  I_a^{(n)} = [b_a, \,\, 3^{-n} + b_a]
  \text{ as in \eqref{eq:RealCantorDefinition}}
  .
  \label{eq:CantorAnnulus}
\end{equation}

\begin{proposition}\label{prop:DiskCantorGeometricProperties}
The $n$-th disk Cantor set $C_{n}(R;\D)$ is a disjoint union of $2^{n}$ annuli:
\[
  C_{n}(R;\mathbb{D})
  = \biguplus_{a \in \Omega^{(n)}}
      D_{a}^{(n)}(0,R)
  \qquad \text{with} \qquad
  D_{a}^{(n)}(0,R) \text{ as in \eqref{eq:CantorAnnulus}}
  .
\]
Each annulus has hyperbolic measure $\mu_{\D}(D_{a}^{(n)}(0,R)) = \frac{R}{3^{n}}$.
Therefore,
\[
  \mu_{\mathbb{D}}\left( C_{n}(R;\mathbb{D})\right)
  = \left( \frac{2}{3} \right)^{n} \cdot R
  .
\]
\end{proposition}

\begin{proof}
Recall the definition \eqref{eq:HyperbolicMeasure} of the hyperbolic measure $\mu_{\D}$.
Next, given an arbitrary measurable function $F : [0,\infty) \to [0,\infty]$,
introduce polar coordinates, set $s = r^{2}$ and then $t = \frac{s}{1-s}$ to yield
\begin{align*}
  \int_{\D}
    F\!\left( \frac{|z|^{2}}{1-|z|^{2}}\right)
  d \mu_{\D}(z)
  & = 2 \int_{0}^{1}\!
          F\left( \frac{r^{2}}{1-r^{2}}\right)
          \frac{r}{(1-r^{2})^{2}}
        \,d r \\
  & = \int_{0}^{1}\!
        F\left( \frac{s}{1-s}\right)
        (1-s)^{-2}
      \,ds
    = \int_{0}^{\infty} F(t) \,d t
  .
\end{align*}
Since $C_{n}(R;\R) = \biguplus_{a \in \Omega^{(n)}} R \cdot I_{a}^{(n)}$,
where the intervals $I_{a}^{(n)} = [b_a, 3^{-n} + b_a] \subset [0,\infty)$
have length $|I_{a}^{(n)}| = 3^{-n}$, then
\[
  C_{n}(R;\D)
  = \biguplus_{a \in \Omega^{(n)}}
      D_{a}^{(n)}(0,R)
  = \biguplus_{a \in \Omega^{(n)}}
      \big\{
        w \in \D
        :
        \varphi^{-1}(R \cdot b_a)
        \leq |w|
        \leq \varphi^{-1}(R \cdot (3^{-n} + b_a))
      \big\}
\]
is a disjoint union of $2^{n}$ annuli, each of which has $\mu_{\D}$-measure
\begin{align*}
  \mu_{\mathbb{D}}\left( D_{a}^{(n)}(0,R)\right)
  & = \int_{\D}
        1_{R \cdot I_{a}^{(n)}}
        \big( |z|^{2}/(1-|z|^{2}) \big)
      \, d \mu_{\D}(z) \\
  & = \int_{0}^{\infty}
        1_{R \cdot I_{a}^{(n)}}(t)
      \, d t
    = \frac{\,R}{3^{n}}.
  \qedhere
\end{align*}
\end{proof}

\section{Localization to the Cantor set in the disk}
\label{sec:CantorLocalizationOperators}

In this section, we study the spectral properties of the time-scale
localization operators
\[
  P_{C_{n}(R;\mathbb{D})}^{(\alpha)}
  \quad \text{for} \quad
  R\in (0,\infty), \,\, n\in \mathbb{N},
  \text{ and } \alpha \in (0,\infty),
\]
as defined in \eqref{eq:OurLocalizationOperator}.

\subsection{Eigenvalues of the Cantor-type localization operator}

\begin{proposition}\label{prop:CantorEigenvaluesExplicit}
Let $\alpha, R \in (0,\infty)$ and $n \in \N$.
Then the eigenvalues of the operator $P_{C_{n}(R;\D)}^{(\alpha)}$ introduced in
Equation~\eqref{eq:OurLocalizationOperator} are given by
\[
  \lambda_{k}
  = \lambda_{k}^{(\alpha)}\left(C_{n}(R;\D)\right)
  = \int_{0}^{\infty}
      g_{k}(t;\alpha) \, 1_{C_{n}(R;\R)}(t)
    \, d t
  \quad \text{for} \quad
  k \in \N_0
  ,
\]
where $g_{k}(\cdot \,;\alpha) : (0,\infty) \to (0,\infty)$ is given by
\begin{equation}
  g_{k}(t;\alpha)
  = [B(k+1,2\alpha)]^{-1}
    \cdot \left( \frac{t}{1+t}\right)^{k}
    \cdot (1+t)^{-(1+2\alpha)}
  .
  \label{eq:GkDefinition}
\end{equation}
\end{proposition}

\begin{proof}
Observe that the orthogonality of the basis functions $e_{k}^{\alpha }(z) = \sqrt{\gamma_{k}} z^{k}$
in the reproducing kernel expansion \eqref{repBergman} holds for any radial measure
$\mu (\left\vert z\right\vert)dA(z)$, with $\mu : [0,\infty) \rightarrow [0,\infty)$ satisfying
$\int_{0}^{1}\mu (r)dr < \infty$.
This follows from the following calculation, for $k \neq \ell$:
\begin{align*}
  \int_{\mathbb{D}}
    z^{k} \,
    \overline{z^{\ell}} \,
    \mu (|z|)
  d A(z)
  &= \frac{1}{\pi}
     \int_{0}^{1}
       \int_{0}^{2\pi}
         (r e^{i\theta })^{k}
         (r e^{-i\theta})^{\ell}
         \mu (|re^{i\theta}|)
         r
       d \theta
     d r \\
  &= \frac{1}{\pi}
     \int_{0}^{1}
       r^{k + \ell + 1}
       \mu (r)
       \int_{0}^{2\pi}
         e^{i\theta (n-m)}
       d \theta
     d r
   = 0
  .
\end{align*}

Thus, we can set
\[
  \mu (|z|)
  = 2 \alpha
    \cdot (1 - | z|^{2})^{2\alpha -1}
    \cdot 1_{C_{n}(R;\D)}(z)
  = 2 \alpha
    \cdot (1 - | z|^{2})^{2\alpha -1}
    1_{C_{n}(R;\R)} \left( \frac{|z|^{2}}{1-|z|^{2}}\right)
\]
and combine \eqref{repBergman} with the orthogonality of the $e_{k}^{\alpha }(z)$
with respect to $\mu (|z|)dA(z)$ to obtain%
\footnote{Here, we use that the series for $\mathcal{K}_{\D}^{\alpha}(\cdot,w)$
given in \eqref{repBergman} is a power series which is convergent on $\D$
and thus converges locally on $\D$.
Since the measure $\mu$ has compact support in $\D$ (since $C_n(R;\D) \subset \D$ is compact),
this allows to interchange the series with the integral.}
\[
  \bigl[P_{C_{n}(R;\D)}^{(\alpha)} (e_{k}^{\alpha })\bigr] (w)
  = \int_{\D}
      e_{k}^{\alpha }(z)
      \overline{\mathcal{K}_{\mathbb{D}}^{\alpha }(z,w)}\,
      \mu (|z|)
    d A(z)
  = \left[
      \int_{\D}
        |e_{k}^{\alpha}(z)|^{2}
        \, \mu (|z|)
      d A(z)
    \right]
  e_{k}^{\alpha}(w)
  .
\]
Thus, $e_{k}^{\alpha }$ is an eigenfunction of $P_{C_{n}(R;\D)}^{(\alpha )}$ with eigenvalue
\begin{align*}
  \lambda _{k}^{(\alpha )}(C_{n}(R;\D))
  & = \int_{\D}
        |e_{k}^{\alpha}(z)|^{2}\,
        \mu (|z|)
      d A(z) \\
  & = 2 \alpha \gamma_{k}
      \int_{\D}
        |z|^{2k}
        \cdot (1 - |z|^{2})^{2\alpha - 1}
        \cdot 1_{C_{n}(R;\R)} \left( \frac{|z|^{2}}{1-|z|^{2}} \right)
      d A(z) \\
  &= 4 \alpha \gamma_{k}
     \int_{0}^{1}
       1_{C_{n}(R;\R)}\left( \frac{r^{2}}{1-r^{2}}\right)
       \cdot r^{2k+1}
       \cdot (1 - r^{2})^{2\alpha -1}
     \,dr
  .
\end{align*}
Using the substitutions $s = r^{2}$ and $t = \frac{s}{1-s} = \frac{1}{1-s} - 1$, we finally see
\begin{align*}
  \lambda_{k}^{(\alpha)}(C_{n}(R;\D))
  & = 2\alpha \gamma_{k}
      \int_{0}^{1}
        1_{C_{n}(R;\R)}\left( \frac{s}{1-s}\right)
        \cdot s^{k}
        \cdot (1 - s)^{2\alpha -1}
      \, d s \\
  & = 2 \alpha \gamma_{k}
      \int_{0}^{\infty}
        1_{C_{n}(R;\R)}(t)
        \cdot \left(\frac{t}{1+t}\right)^{k}
        \cdot (1 + t)^{-(2\alpha + 1)}
      \, d t \\
  & = \int_{0}^{\infty}
        1_{C_{n}(R;\R)}(t)
        \cdot g_{k}(t;\alpha)
      \, d t
  .
  \qedhere
\end{align*}
\end{proof}

\begin{remark}
We observe that the function $g_{k}(\cdot \,;\alpha )$ is the density
function of the \emph{Beta prime distribution} $\beta ^{\prime }(k+1,2\alpha)$
with form parameters $k+1$ and $2\alpha $;
see \cite[Equation~(25.79)]{JohnsonContinuousUnivariateDistributionsVolTwo}.
In particular, this implies
\begin{equation}
  \int_{0}^{\infty}
    g_{k}(t;\alpha )
  \,dt
  = 1
  \label{eq:gkIsDensityFunction}
\end{equation}
and the following probabilistic interpretation of the eigenvalues
\[
  \lambda_{k}^{(\alpha)}(C_{n}(R;\mathbb{D}))
  = \mathbb{E} \left( 1_{C_{n}(R;\R)}(X) \right)
  \quad \text{if } X \sim \beta^{\prime}(k+1,2\alpha)
  .
\]
Recalling the definition of $g_{k}(\cdot \,;\alpha )$ and applying
\cite[Section~3.197.8]{TableOfIntegralsAndSeries} (with $\mu =1$, $\nu = k + 1$, $\alpha =1$,
$\lambda =-(k+1+2\alpha )$ and $u = y$) shows that the cumulative distribution
function of $\beta^{\prime }(k+1,2\alpha)$ is given by
\[
  F_{k,\alpha}(y)
  = \int_{0}^{y}
      g_{k}(t;\alpha)
    \, d t
  = \frac{y^{1+k}}{(1+k)
    \cdot B(1+k,2\alpha)}
    \cdot {}_{2}F_{1}(k+1+2\alpha ,k+1;k+2;-y)
  ,
\]
in terms of the ordinary hypergeometric function
${}_{2}F_{1}( a,b,c;z ) = \sum_{j=0}^{\infty} \frac{(a)_{j}(b)_{j}}{j!(c)_{j}}z^{j}$.
When $n = 0$, then $1_{C_{0}(R;\R)} = 1_{[0,R]}$ and thus
\[
  \lambda_{k}^{(\alpha)}(C_{0}(R;\D))
  = \frac{R^{1+k}}{(1+k)
    \cdot B(1+k,2\alpha)}
    \cdot {}_{2}F_{1}(k+1+2\alpha ,1+k;k+2;-R)
  .
\]
\end{remark}

\subsection{Upper bounding the eigenvalues}
\label{sub:CantorLocalizationUpperBounds}

The upper bounds for the eigenvalues depend on the following pointwise
estimate for the density functions $g_{k}(\cdot \,;\alpha )$.
The details of the proof are in Section~\ref{sec:DensityFunctionUpperBound}.

\begin{lemma}\label{lem:DensityFunctionUpperBound}
For each $\alpha \in (0,\infty)$, there is a constant $C = C(\alpha) > 0$ satisfying
\begin{equation*}
  g_{k}(x;\alpha)
  \leq C \cdot (1+x)^{-1}
  \qquad \forall \,x\in (0,\infty )\text{ and } k \in \mathbb{N}_{0}.
\end{equation*}
\end{lemma}

In this section, we prove the following upper bound
for the eigenvalues $\lambda_{k}^{(\alpha)}(C_{n}(R;\D))$:

\begin{proposition}\label{prop:EigenvaluesUpperBound}
For each $\alpha \in (0,\infty)$, there is a constant $C = C(\alpha) > 0$ satisfying
\[
  \lambda_{k}^{(\alpha)}(C_{n}(R;\D))
  \leq C
       \begin{cases}
       (2/3)^{n} \cdot R,                         & \text{if } 0 < R \leq 1, \\
       (2/3)^{n} \cdot R^{1-\frac{\ln 2}{\ln 3}}, & \text{if }1 \leq R \leq 3^{n}, \\
       1,                                         & \text{if }R \geq 3^{n}
       \end{cases}
\]
for all $k \in \mathbb{N}_{0}$, $n \in \mathbb{N}$, and $R \in (0,\infty)$.
\end{proposition}

\begin{proof}
First of all, recall from Equation~\eqref{eq:gkIsDensityFunction} that
each $g_{k}(\cdot \,;\alpha )$ is a probability density function on $(0,\infty )$, so that
\[
  \lambda_{k}^{(\alpha )}(C_{n}(R;\mathbb{D}))
  = \int_{C_{n}(R;\R)}
      g_{k}(y;\alpha )
    \, d y
  \leq \int_{0}^{\infty}
         g_{k}(y;\alpha )
       \, d y
  =    1.
\]
This establishes the desired estimate in case of $R \geq 3^{n}$.
Next, recall that the Lebesgue measure $\lambda (C_{n}(R;\mathbb{R}))$ of the Cantor set
$C_{n}(R;\R)$ is $\lambda (C_{n}(R;\R)) = (2/3)^{n} \cdot R$.
Furthermore, Lemma~\ref{lem:DensityFunctionUpperBound} yields a constant
$C = C(\alpha) > 0$ satisfying $g_{k}(x;\alpha )\leq C\cdot (1+x)^{-1} \leq C$
for all $x\in (0,\infty )$ and $k\in \mathbb{N}_{0}$.
Therefore,
\[
  \lambda_{k}^{(\alpha)}(C_{n}(R;\D))
  = \int_{C_{n}(R;\R)}
      g_{k}(y;\alpha )
    \, d y
  \leq C \cdot \lambda (C_{n}(R;\R))
  =    C \cdot (2/3)^{n} \cdot R.
\]
This proves the desired estimate for the case that $0 < R \leq 1$.
Finally, let us consider the case $1 \leq R \leq 3^{n}$.
Due to this assumption on $R$, there is $t\in \{0,\dots ,n-1\}$ such that
\[
  3^{t}
  \leq R
  \leq 3^{t+1}
  \quad \text{and hence} \quad
  \frac{\ln R}{\ln 3}-1
  \leq t
  \leq \frac{\ln R}{\ln 3}.
\]
Using the representation \eqref{eq:RealLineCantorSet} of the Cantor set
$C_{n}(R;\R)$ and $\Omega ^{(n)}\setminus \{0\} = \biguplus_{m=1}^{n}\Omega _{m}^{(n)}$, we see
\[
  \lambda_{k}^{(\alpha)}(C_{n}(R;\D))
  = \int_{0}^{R/3^{n}}
      g_{k}(y;\alpha )
    \, d y
    + \sum_{m=1}^{n}\,
        \sum_{a\in \Omega _{m}^{(n)}}\,
          \int_{R \cdot b_{a}}^{R\cdot (b_{a}+3^{-n})}
            g_{k}(y;\alpha)
          \, d y.
\]
First, note as a consequence of $g_{k}(x;\alpha) \leq C\cdot (1+x)^{-1}\leq C$ that
$\int_{0}^{R/3^{n}}g_{k}(y;\alpha )\,dy\leq C\frac{R}{3^{n}}$.
Next, note that if $a\in \Omega _{m}^{(n)}$, then \eqref{eq:LowerCantorBoundAsymptotic}
implies $b_{a} \geq 2/3^{m} \geq 3^{-m}$ and $R \cdot b_{a}\geq 3^{t-m}$.
This implies for $y\in [R \cdot b_{a}, R \cdot (b_{a}+3^{-n})]$ that
\[
  (1+y)^{-1}
  \leq (1+3^{t-m})^{-1}
  \leq 3^{-(t-m)_{+}}
  \quad \text{where}\quad
  x_{+} := \max \{0, x\},
\]
and hence
\[
  \int_{R\cdot b_{a}}^{R\cdot (b_{a}+3^{-n})}
    g_{k}(y;\alpha )
  \, d y
  \leq C \int_{R\cdot b_{a}}^{R\cdot (b_{a}+3^{-n})}
           (1+y)^{-1}
         \, d y
  \leq C \cdot \frac{R}{3^{n}} \cdot 3^{-(t-m)_{+}}.
\]
Recalling, from \eqref{eq:SpecialIndexSetCardinality} that $\# \Omega_{m}^{(n)} = 2^{n-m}$,
we conclude
\[
  \sum_{m=1}^{n}
    \sum_{a\in \Omega _{m}^{(n)}}
      \int_{R\cdot b_{a}}^{R\cdot (b_{a}+3^{-n})}
        g_{k}(y;\alpha)
      \, d y
  \leq C \cdot (2/3)^{n} \cdot R \cdot 2^{-t}
       \sum_{m=1}^{n}
         2^{t-m}3^{-(t-m)_{+}}
  .
\]
Introducing the new summation index $\ell = t - m$, we see that
\[
  \sum_{m=1}^{n}
    2^{t-m}
    3^{-(t-m)_{+}}
  \leq \sum_{\ell \in \mathbb{Z}}
         2^{\ell}
         3^{-\ell_{+}}
  =    \sum_{\ell = -\infty}^{-1}
         2^{\ell}
       + \sum_{\ell = 0}^{\infty}
           (2/3)^{\ell}
  =    1 + \frac{1}{1-\frac{2}{3}}
  =    4.
\]
Furthermore, since $t \geq \frac{\ln R}{\ln 3} - 1$,
we have $2^{-t} \leq 2 \cdot 2^{-\ln R/\ln 3} = 2 \cdot R^{-\ln 2/\ln 3}$.
By combining the estimates that we collected, we finally conclude that
\[
  \lambda_{k}^{(\alpha)}\left( C_{n}(R;\D)\right)
  \leq C \frac{R}{3^{n}}
       + 8 C \cdot (2/3)^{n} \cdot R^{1-\frac{\ln 2}{\ln 3}}.
\]
It remains to observe that $R/3^{n} \leq (2/3)^{n} \cdot R^{1-\frac{\ln 2}{\ln 3}}$,
which easily follows from the condition $R\leq 3^{n}$.
\end{proof}

\subsection{Lower bounding the first eigenvalue}
\label{sub:CantorLocalizationLowerBounds}

In this subsection we prove that the eigenvalue
$\lambda_{0}^{(\alpha)}(C_{n}(R;\D))$ of $P_{C_{n}(R;\mathbb{D})}^{(\alpha)}$ fulfills a
lower bound which matches the upper bound from Proposition~\ref{prop:EigenvaluesUpperBound}.

\begin{proposition}\label{prop:EigenvaluesLowerBound}
For each $\alpha \in (0,\infty)$, there is a constant $C = C(\alpha) > 0$ satisfying
\[
  \lambda _{0}^{(\alpha )}(C_{n}(R;\mathbb{D}))
  \geq C
       \begin{cases}
         (2/3)^{n} \cdot R,                         & \text{if }0 < R \leq 1, \\
         (2/3)^{n} \cdot R^{1-\frac{\ln 2}{\ln 3}}, & \text{if }1 \leq R \leq 3^{n}, \\
         1,                                         & \text{if }R \geq 3^{n}
       \end{cases}
\]
for all $n\in \mathbb{N}$ and $R\in (0,\infty)$.
\end{proposition}

\begin{proof}
In case of $R\geq 3^{n}$, we have $C_{n}(R;\R) \supset [0,R/3^{n}] \supset [0,1]$, and hence
\[
  \lambda_{0}^{(\alpha)}(C_{n}(R;\D))
  \geq \int_{0}^{1}
         g_{0}(y;\alpha)
       \, d y
  =:   C_{1}(\alpha)
  >    0,
\]
since $g_{0}(\cdot \,;\alpha)$ is a positive continuous function.
Likewise, if $0 < R \leq 1$, then $C_{n}(R;\R) \subset [0,1]$.
Since the continuous, positive function $g_{0}(\cdot \,;\alpha)$ is lower bounded on
the compact set $[0,1]$ (say, $g_{0}(x;\alpha) \geq C_{2}$ for $x \in [0,1]$
with $C_{2} = C_{2}(\alpha) > 0$), we thus see
\[
  \lambda_{0}^{(\alpha)}(C_{n}(R;\D))
  = \int_{C_{n}(R;\R)}
      g_{0}(y;\alpha)
    \, d y
  \geq C_{2}
       \int_{C_{n}(R;\mathbb{R})}
         1
       \, d y
  =    C_{2} \cdot (2/3)^{n} \cdot R,
\]
proving the desired bound for the case $0 < R \leq 1$.
Finally, consider the case $1 \leq R \leq 3^{n}$.
For brevity, define $C_{3} := [B(1,2\alpha)]^{-1}$.
Choose $t\in \{0,\dots ,n-1\}$ such that $3^{t} \leq R \leq 3^{t+1}$,
whence $\frac{\ln R}{\ln 3} - 1 \leq t \leq \frac{\ln R}{\ln 3}$.
Using the representation \eqref{eq:RealLineCantorSet} of the Cantor set, and observing
that $\Omega_{t+1}^{(n)} \subset \Omega^{(n)}$ since $t+1\in \{1,\dots ,n\}$, it follows that
\[
  \lambda_{0}^{(\alpha)}(C_{n}(R;\D))
  \geq C_{3}
       \sum_{a\in \Omega_{t+1}^{(n)}}
         \int_{R\cdot b_{a}}^{R\cdot (b_{a}+3^{-n})}
           (1+y)^{-(1+2\alpha)}
         \, d y
  .
\]
Now, note that if $a\in \Omega _{t+1}^{(n)}$ and $y\leq R\cdot (b_{a}+3^{-n})$,
then Equation~\eqref{eq:LowerCantorBoundAsymptotic} shows that
\[
  y
  \leq R \cdot (b_{a}+3^{-n})
  \leq R \cdot (3\cdot 3^{-(t+1)}+3^{-n})
  \leq 4 R \cdot 3^{-(t+1)}
  \leq 4,
\]
whence $(1+y)^{-(1+2\alpha )} \!\geq\! 5^{-(1+2\alpha )} =: C_{4}$.
Next, recall from Equation~\eqref{eq:SpecialIndexSetCardinality}
that $\# \Omega_{t+1}^{(n)} = 2^{n-t-1}$.
Overall, we thus see as desired that
\begin{align*}
  \lambda_{0}^{(\alpha)}(C_{n}(R;\D))
  & \geq C_{3} C_{4} \cdot 2^{n-t-1} \frac{R}{3^{n}}
  = \frac{C_{3}C_{4}}{2}
    \cdot (2/3)^{n}
    \cdot R
    \cdot 2^{-t} \\
  & \geq \frac{C_{3}C_{4}}{2}
         \cdot (2/3)^{n}
         \cdot 2^{-\frac{\ln R}{\ln 3}}
         \cdot R
    =    \frac{C_{3}C_{4}}{2}
         \cdot (2/3)^{n}
         \cdot R^{1-\frac{\ln 2}{\ln 3}}
  .
  \qedhere
\end{align*}
\end{proof}

\subsection{Proof of Theorem \ref{thm:CantorDiskLocalizationOperatorNorm}}

Since $P_{C_n(R;\D)}^{(\alpha)} : L^2(d A_\alpha) \to L^2(d A_\alpha)$
is self-adjoint, with $P_{C_n(R;\D)}^{(\alpha)} = 0$ on the orthogonal complement
of $\mathrm{span} \{ e_k^{\alpha} \colon k \in \N_0 \}$, we see by combining
Propositions~\ref{prop:CantorEigenvaluesExplicit},
\ref{prop:EigenvaluesUpperBound} and \ref{prop:EigenvaluesLowerBound} that
\begin{equation}
  C_{1}
  \begin{cases}
    (\frac{2}{3})^{n} \cdot R,                         & \text{if }0 < R \leq 1 \\
    (\frac{2}{3})^{n} \cdot R^{1-\frac{\ln 2}{\ln 3}}, & \text{if }1 \leq R \leq 3^{n} \\
    1,                                                 & \text{if }R \geq 3^{n}
  \end{cases}
  \leq \left\Vert P_{C_{n}(R;\mathbb{D})}^{(\alpha )}\right\Vert_{op}
  \leq C_{2}
       \begin{cases}
         (\frac{2}{3})^{n} \cdot R,                         & \text{if } 0 < R \leq 1, \\
         (\frac{2}{3})^{n} \cdot R^{1-\frac{\ln 2}{\ln 3}}, & \text{if } 1 \leq R \leq 3^{n}, \\
         1,                                                 & \text{if } R \geq 3^{n}.
       \end{cases}
  \label{bound1}
\end{equation}
The condition $3^{n}\asymp R^{2}$ implies (for $R \to \infty$) that $1 \ll R \asymp 3^{n/2} \ll 3^n$.
The same condition also implies that $R^{\frac{\ln 2}{\ln 3}}\asymp 2^{n/2}$,
and thus $R^{2 \frac{\ln 2}{\ln 3}} \asymp 2^{n}$.
This implies for $R \to \infty$ that
\[
  (2/3)^{n} \cdot R^{1-\frac{\ln 2}{\ln 3}}
  \asymp (2/3)^{n/2}
  \asymp R^{-1+\frac{\ln 2}{\ln 3}}
\]
and the result follows from \eqref{bound1}.
\hfill$\square$

\section{Proof of Lemma~\ref{lem:DensityFunctionUpperBound}}
\label{sec:DensityFunctionUpperBound}

Recall that the goal is to prove that there exists a constant $C = C(\alpha) > 0$ satisfying
\[
  g_k(x; \alpha) \leq C \cdot (1 + x)^{-1}
  \qquad \forall \, x \in (0,\infty) \text{ and } k \in \N_0 ,
\]
where the function
\[
  g_k (x; \alpha)
  = [B(k+1,2\alpha)]^{-1}
    \cdot \left( \frac{x}{1+x}\right)^{k}
    \cdot (1+x)^{-(1+2\alpha)}
\]
was defined in Equation~\eqref{eq:GkDefinition}.

The proof will be given in three steps.
All implied constants will either be absolute constants
or constants that only depend on $\alpha$.

\medskip{}

\textbf{Step 1} \emph{(Estimating the Beta function):}
By definition, the Beta function can be written in terms of the Gamma function as
$B(x,y) = \frac{\Gamma (x)\,\Gamma (y)}{\Gamma (x+y)}$ for $x,y > 0$.
Furthermore, a precise form of Stirling's formula (see \cite{JamesonStirlingFormula}) shows
that the Gamma function satisfies
\[
  \Gamma (x)
  = \sqrt{2\pi /x} \cdot (x/e)^{x} \cdot e^{\mu (x)}
  \quad \text{for} \quad
  x > 0
  \quad \text{where} \quad
  0\leq \mu (x)\leq (12\,x)^{-1}
\]
and therefore
\[
  \Gamma(x)
  \asymp x^{-1/2} \cdot (x/e)^x
  \qquad \forall \, x \geq 1 ,
\]
which is the only case that we will need.
Indeed, using this estimate, we see for $k \geq 1$ that
\begin{align*}
  [B(k+1,2\alpha)]^{-1}
  & =\frac{\Gamma (k+1+2\alpha )}{\Gamma (2\alpha)\,\Gamma (1+k)} \\
  & \asymp \frac{
             \sqrt{1/(k+1+2\alpha)}
             \cdot \left( \frac{k+1+2\alpha }{e} \right)^{k+1+2\alpha }
           }{
             \Gamma (2\alpha)
             \cdot \sqrt{1/(k+1)}
             \cdot \left( \frac{k+1}{e}\right)^{k+1}
           } \\
  & \lesssim \sqrt{\frac{k+1}{k+1+2\alpha }}
             \cdot \left( \frac{k+1+2\alpha }{e} \right)^{2\alpha}
             \cdot \left( 1+\frac{2\alpha }{k+1} \right)^{k+1} \\
  & \overset{(\ast)}{\lesssim} \left( \frac{k+1+2\alpha }{e} \right)^{2\alpha}
    \lesssim k^{2 \alpha} \quad \text{since } k \geq 1
  .
\end{align*}
Here, the step marked with $(\ast)$ used the well-known fact that
$(1 + \frac{x}{k})^k \xrightarrow[k\to\infty]{} e^x$
for all $x \in \R$, and thus $(1 + \frac{2\alpha}{k+1})^{k+1} \lesssim 1$,
with the implied constant only depending on $\alpha$.

We remark that the estimate for the Beta function that we derived in this step is probably well-known.
We nevertheless decided to give the relatively easy proof since
we could not locate a handy reference.

\medskip {}

\textbf{Step 2} \emph{(Estimating $k^{2\alpha }\cdot y^{k}$):}
Let $y \in (0,1)$ be fixed, and define
\[
  f : \quad
  [0,\infty) \to [0, \infty ), \quad
  t \mapsto t^{2\alpha} \cdot y^{t}
  .
\]
Then $f$ is differentiable on $(0, \infty)$ with derivative
\[
  f^{\prime }(t)
  = 2\alpha \, t^{2\alpha - 1} y^{t}
    + t^{2\alpha } \, \ln (y) \, y^{t}
  = t^{2\alpha -1} \cdot y^{t}(2\alpha +t\ln (y))
  .
\]
From this, it follows that if we define
$t_{0} := -\frac{2\alpha }{\ln y}= \frac{2\alpha }{\ln (1/y)}\in (0,\infty )$,
then $f^{\prime }(t) > 0$ for $t \in (0,t_{0})$ and $f^{\prime }(t) < 0$
for $t\in (t_{0}, \infty)$.
Hence, $f$ has a global maximum in $t = t_{0}$, showing that
\[
  t^{2\alpha }y^{t}
  = f(t)
  \leq f(t_{0})
  =    \left( \frac{2\alpha }{\ln (1/y)}\right)^{2\alpha}
       y^{-\frac{2\alpha }{\ln (y)}}
  =    \left( \frac{2\alpha }{\ln (1/y)} \right)^{2\alpha}
       e^{-2\alpha}
  \lesssim [\ln(1/y)]^{-2\alpha}
\]
for all $t \in [0,\infty)$ and $y \in (0,1)$.

\medskip{}

\textbf{Step 3} \emph{(Completing the proof):}
If we apply the estimate from the preceding step for $t = k$ and $y = \frac{x}{1+x} \in (0,1)$
(where $x \in (0,\infty )$), then we see that
\[
  k^{2\alpha} \cdot \left( \frac{x}{1+x}\right)^{k}
  \lesssim \left[ \ln \frac{1+x}{x} \right]^{-2\alpha}
  =        [\ln (1+x^{-1})]^{-2\alpha}
  \quad \forall \,x \in (0,\infty).
\]
Now, for $x \geq 1$ note that $1 + x^{-1} \leq 2$, and hence
\[
  \ln (1+x^{-1})
  = \int_{1}^{1+x^{-1}}
      t^{-1}
    \, d t
  \geq \frac{1}{2}\,x^{-1}
\]
from which we see---because of $\alpha >0$---that
\[
  [\ln (1+x^{-1})]^{-2\alpha}
  \leq 2^{2\alpha} \cdot x^{2\alpha}
  \leq 2^{2\alpha} \cdot (1+x)^{2\alpha }.
\]
Likewise, if $0 < x \leq 1$ then $\ln (1+x^{-1}) \geq \ln (2)$ and hence
\[
  [\ln (1+x^{-1})]^{-2\alpha }
  \leq [ \ln (2) ]^{-2\alpha}
  \leq [ \ln (2)]^{-2\alpha} \cdot (1+x)^{2\alpha}
  .
\]
All in all, we have thus shown that
\[
  k^{2\alpha} \cdot \left( \frac{x}{1+x}\right)^{k}
  \lesssim (1+x)^{2\alpha}
  \qquad \forall \,x \in [0,\infty) \text{ and } k \in \mathbb{N}_{0}.
\]
Combining this with the estimate from Step~1, we see for $k \geq 1$ that
\begin{align*}
  g_{k}(x;\alpha )
  & = [B(k+1,2\alpha)]^{-1}
      \cdot \left( \frac{x}{1+x}\right)^{k}
      \cdot (1+x)^{-(1+2\alpha)} \\
  & \lesssim k^{2\alpha}
             \cdot \left( \frac{x}{1+x}\right)^{k}
             \cdot (1+x)^{-2\alpha }
             \, (1+x)^{-1}
     \lesssim (1+x)^{-1}.
\end{align*}
Finally, in case of $k = 0$, we see directly from the definition of the Beta function that
\[
  B(k+1,2\alpha)
  = B(1,2\alpha)
  = \int_{0}^{1}
      (1-t)^{2\alpha-1}
    \, d t
  = \int_{0}^{1}
      s^{2\alpha -1}
    \, d s
  = \frac{1}{2\alpha}
  ,
\]
and hence
\[
  g_{0}(x;\alpha )
  = 2\alpha \cdot (1+x)^{-(1+2\alpha )}
  \leq 2\alpha \cdot (1+x)^{-1},
\]
since $1+x \geq 1$ and $1+2\alpha \geq 1$.
\hfill$\square$

\begin{acknowledgement*}
We thank Helge Knutsen for several comments and corrections on an earlier
version of the manuscript.
\end{acknowledgement*}

\end{document}